\documentclass[pra,aps,superscriptaddress,showpacs,floatfix,tightenlines,twocolumn]{revtex4-1}

\usepackage{graphicx}
\usepackage{dcolumn}   
\usepackage{bm}        
\usepackage{amssymb}   
\usepackage{amsmath}
\usepackage{url}
\usepackage{amsthm}
\usepackage{layout}
\usepackage{epsfig}
\usepackage{graphicx}
\usepackage{epstopdf}
\usepackage{booktabs}
\usepackage{float}
\usepackage[hyperindex,breaklinks]{hyperref}
\newcommand{\ket}[1]{|#1\rangle}                  
\newcommand{\bra}[1]{\langle #1 |}     
\newcommand{\inner}[2]{\langle #1|#2\rangle}
\newcommand{\proj}[1]{|#1\rangle \langle #1|}
\newcommand{\cs}[2]{|\langle #1|#2\rangle|^2}  	  %
\def\Tr{{\rm Tr}}

\newtheorem{theorem}{Theorem}
\newtheorem{lemma}{Lemma}
\newtheorem{corollary}{Corollary}

\begin{document}
\title{Entropic uncertainty relations for multiple measurements}

\author{Shang Liu}
\affiliation{School of Physics, Peking University, Beijing 100871, China}
\author{Liang-Zhu Mu}
\affiliation{School of Physics, Peking University, Beijing 100871, China}
\author{Heng Fan}
\email{hfan@iphy.ac.cn}
\affiliation{Institute of Physics, Chinese Academy of Science, Beijing 100190, China}
\affiliation{Collaborative Innovation Center of Quantum Matter, Beijing 100190, China}

\begin{abstract}
We present the entropic uncertainty relations for
multiple measurement settings in quantum mechanics.
Those uncertainty relations are obtained for both cases
with and without the presence of quantum memory. They
take concise forms which can be proven in a unified method and easy to calculate.
Our results recover the well known entropic uncertainty relations
for two observables, which show the uncertainties about the outcomes of
two incompatible measurements.
Those uncertainty relations are applicable in both
foundations of quantum theory and the security of many
quantum cryptographic protocols.

\end{abstract}
\maketitle

\emph{Introduction.}---Uncertainty principle is one unique feature of quantum mechanics differing
from the classical case.  Heisenberg \cite{Heisenberg} formulated the first uncertainty relation
which shows that one cannot predict the outcomes with arbitrary precision for two
incompatible measurements simultaneously, such as position and momentum, on a particle.
As a fundamental property, uncertainty principle is continuously attracting lots of attention and research interests.
Variants of uncertainty relations are presented in the past years. One type of best known uncertainty relations
today is in the form proposed by Robertson \cite{Robertson}.
For arbitrary two observables $U$ and $V$, the uncertainty relation given by Robertson takes the form,
$\sigma_U\sigma_V\geq \left|\bra{\psi}\frac{1}{2\text{i}}[U,V]\ket{\psi}\right|$,
where $\sigma$ is the standard deviation of an observable.
This bound of uncertainty, however, may have the drawback of being state-dependent.
So for some states, this bound is trivial.
Deutsch \cite{Deutsch} proposed to use Shannon entropy
 ($H(\{p_i\}):=-\sum_ip_i\log p_i$, the base of logarithm is assumed being 2 hereafter) as a proper measure of uncertainty
and presented the entropic uncertainty relation:
\begin{equation}
H(U)+H(V)\geq -2\log\left(\frac{1+\sqrt{c(U,V)}}{2}\right)
\end{equation}
Here, $U$ and $V$ are two projective measurements with bases $\{\ket{u_i}\}$ and $\{\ket{v_j}\}$ respectively.
In this form, the uncertainty is naturally quantified by entropy in the
information-theoretical context instead of the standard deviation.
Also, we use the notations $c(u_i,v_j):=\cs{u_i}{v_j}$, $c(U,V):=\max_{i,j}c(u_i,v_j)=\max_{i,j}|\inner{u_i}{v_j}|^2$, which are consistent with past works.
We may find that the bound of uncertainty depends only on the complementarity of the observables avioding the shortcomings
of state-dependent.
Maassen and Uffink \cite{Maassen&Uffink} (MU) further strengthened Deutsch's inequality to a tighter and more succinct form as
\begin{equation}
H(U)+H(V)\geq -\log c(U,V).
\label{MU}
\end{equation}
In this inequality, the largest uncertainty can be obtained for observables which are mutually unbiased, i.e.,
the quantities $c(u_i,v_j):=\cs{u_i}{v_j}$ take the same value, $1/\sqrt {d}$, which depends on the dimension $d$.
It is known that the mutually unbiased bases (MUBs) are useful in quantum information processing, in particular,
for quantum key distributions, see for example Refs. \cite{Cerf-d,XiongPRA,FanPRL,PhysRep,Vwani}.
By considering that the number of MUBs can be at most $d+1$ \cite{Vwani}, other than
only restricting to 2,
it is natural to investigate the uncertainty relation with more than two measurements even for
the simplest two-dimensional case, see FIG. 1.

Many efforts have been made to generalize the uncertainty relations to more than two observables,
see \cite{Wehner&Winter} for a review. Significant progresses have been
made in this direction for case of MUBs \cite{Ivanovic,Sanchez,Sanchez-Ruiz}, and a few latest results in other cases \cite{Friedland,Karol0,Karol,Marco}. We will present the entropic uncertainty relations for multiple
measurements with general condition.

On the other hand, a remarkable result of the uncertainty principle recently is to investigate the
effect of quantum memory which is available with current technologies.
It is shown that the extent of uncertainty can be reduced with the help of memory which
might be entangled with the measured system \cite{Berta}.
This uncertainty relation is confirmed experimentally \cite{experiment1,experiment2} and can be applied
in studying the security of quantum cryptography. Again, the inequality is only for two measurements while
the case of multiple observables is of fundamental interest and of practical applications for quantum
key distributions with more than two measurements settings \cite{XiongPRA,PhysRep}.
We remark that the uncertainty inequality has been extended to multi-partite systems \cite{Coles}
and can be related with many concepts such as teleportation, entanglement witness in quantum information processing \cite{HuPRA,experiment2}. Still, the general uncertainty inequalities for multiple measurements
in the presence of quantum memory are still absent.
In this Letter, we will present the uncertainty inequalities for multiple measurements which
are in a unified framework for both cases with or without the quantum memory.

\begin{figure}
\includegraphics[width=0.4\linewidth]{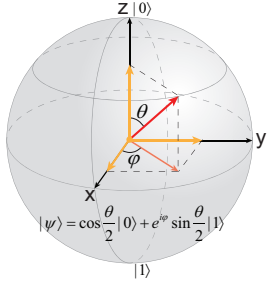}
\caption{(color online) A typical set of three MUBs in two-dimensional Hilbert space is visualized on the Bloch sphere. These measurements can provide a complete description of any quantum state in this space. }
\label{Bloch}
\end{figure}


\emph{R\'{e}nyi entropy and the generalization of Deutsch's inequality. }---The generalization of Deutsch's inequality is relatively simple, but we need the concept of R\'{e}nyi entropy \cite{Renyi} to present our result. For a set of probabilities $\{p_i\}$ and any real number $\alpha>0$, the classical R\'{e}nyi entropy is defined as
\begin{equation}
H_{\alpha}(\{p_i\}):=\frac{1}{1-\alpha}\log(\sum_ip_i^{\alpha}).
\end{equation}
R\'{e}nyi entropy is a monotonic decreasing function with respect to $\alpha$ when the probability distribution is fixed. Taking the limit as $\alpha\rightarrow 1$, one reaches the definition of Shannon entropy: $\lim_{\alpha\rightarrow 1} H_{\alpha}(p)\equiv H_1(p)=-\sum_ip_i\log p_i$, where we have used $p$ as an abbreviation of $\{p_i\}$. On the other hand, we also have $\lim_{\alpha\rightarrow\infty}H_{\alpha}(p)=-\log(\max_ip_i)$. Obviously, having similar properties with Shannon entropy, R\'{e}nyi entropies are also appropriate tools for the description of uncertainty.

Now, suppose that we have $N$ projective measurements $M_1$, $M_2$, $\ldots$, $M_N$ whose bases are $\{\ket{u^1_{i_1}}\}$, $\{\ket{u^2_{i_2}}\}$, $\ldots$, $\{\ket{u^N_{i_N}}\}$, respectively. We have the following theorem.
\begin{theorem}\label{generalization_D}
The following entropic uncertainty relation holds.
\begin{equation}
\sum_{m=1}^N H_{\infty}(M_m)\geq -\log(h),
\end{equation}
where
\begin{equation}
h=\max_{i_1,i_2,...i_N}\prod_{m=1}^{N}(\frac{1+\sqrt{c(u_{i_m}^m,u_{i_{m+1}}^{m+1})}}{2}).
\label{bound_D}
\end{equation}
\end{theorem}
Here, module $N$ is assumed for superscripts. We note that because of the monotonicity of R\'{e}nyi entropy, we can actually replace the l.h.s. by $H_{\alpha_1}(M_1)+H_{\alpha_2}(M_2)+\ldots+H_{\alpha_N}(M_N)$ for an arbitrary set of $\{\alpha_i\}$. The statement in Theorem \ref{generalization_D} is the tightest version. Especially, if we choose all $\alpha_i$'s to be 1, the natural generalization of Deutsch's work is obtained. This is not a simple summation of two-observable inequalities since the maximum is taken outside the multiplication.

Proof of this result is not difficult, but still requires many lines of argument. Roughly speaking, we aim at giving an upper bound for the quantity $p^1_{i_1}p^2_{i_2}...p^N_{i_N}$, where $p^m_{i_m}$ is the probability of getting the $i_m$th result of the $m$th measurement. This quantity can be factorized into terms like $\sqrt{p^m_ip^n_j}=|\inner{\psi}{u^m_i}\inner{\psi}{u^n_j}|$. If we imagine that all the vectors are in real Euclidean space, we simply have,
\begin{eqnarray}
|\inner{\psi}{u^m_i}\inner{\psi}{u^n_j}|\leq \frac{1}{2}(1+|\inner{u^m_i}{u^n_j}|),
\end{eqnarray}
which obviously implies the inequality in this theorem.
However the vectors actually live in a complex Hilbert space, but similar procedure is still able to be applied.
This concludes our proof. More details are given in the supplementary information \cite{sup}.

\emph{Generalization of MU inequality. }---
We now consider the MU bound for multi-observable uncertainty. The state of the measured system is denoted by $\rho$, which is generally a mixed state with its von Neumann entropy defined as, $S(\rho):=-\Tr(\rho\log\rho)$.
\begin{theorem}\label{generalization_MU}
The following entropic uncertainty relation holds,
\begin{equation}
\sum_{m=1}^NH(M_m)\geq-\log(b)+(N-1)S(\rho),
\label{bound_MU}
\end{equation}
where
\begin{equation}
b=\max_{i_N}\{\sum_{i_2\sim i_{N-1}}\max_{i_1}[c(u^1_{i_1},u^2_{i_2})]\Pi_{m=2}^{N-1}c(u^m_{i_m},u^{m+1}_{i_{m+1}})\}.
\end{equation}
\end{theorem}
For example, if $N=3$, we have:
\begin{equation}
b=\max_k\{\sum_{j}\max_i[c(u^1_i,u^2_j)]c(u^2_j,u^3_k)\}.
\end{equation}
The outline of the proof is sketched below, where the method used is inspired by the excellent works of Coles \emph{et al.} \cite{Coles}.

First, one can easily verify that for a projective measurement, say $U$, we have the following relation,
\begin{equation}
H(U)-S(\rho)=S(\rho||\sum_i\proj{u_i}\rho\proj{u_i}),
\end{equation}
where $S(\rho||\sigma):=\Tr(\rho\log\rho)-\Tr(\rho\log\sigma)$ is the quantum relative entropy. Then, by using the well known theorem that quantum channels never increase relative entropy (see Page 208 of Ref. \cite{Vedral} and Ref. \cite{Nielsen&Chuang}), i.e. $S(\rho||\sigma)\geq S(\mathcal{E}(\rho)||\mathcal{E}(\sigma))$ for any trace-preserving operation $\mathcal{E}$, we obtain the following inequality from the equation above,
\begin{equation}
H(U)+H(V)\geq S(\rho||\sum_{i,j}p_ic(u_i,v_j)\ket{v_j}\bra{v_j})+2S(\rho),
\end{equation}
where the operation $\mathcal{E}$ utilized is $\mathcal{E}(\rho):=\sum_{j}\proj{v_j}\rho\proj{v_j}$, and $p_i=\bra{u_i}\rho\ket{u_i}$ is the probability of obtaining the $i$-th outcome of $U$. Notice that the r.h.s. of this inequality again contains a term of relative entropy, thus we can apply the same method continuously and obtain inequalities with arbitrarily many entropic terms in the l.h.s. More precisely, we find,
\begin{equation}
-NS(\rho)+\sum_{m=1}^NH(M_m)\geq S(\rho||\sum_j\beta^N_j\ket{u^N_j}\bra{u^N_j}),
\end{equation}
where $\beta^N_j:=\sum_{i_1,i_2,...i_{N-1}} p^1_{i_1}c(u^1_{i_1},u^2_{i_2})...c(u^{N-1}_{i_{N-1}},u^N_j)$ with $p^1_{i_1}=\bra{u^1_{i_1}}\rho\ket{u^1_{i_1}}$. Finally, by slightly weakening this inequality, we obtain exactly the result shown in Theorem \ref{generalization_MU}.

Let us take a further look of this result. Notice that since the maximum over $j$ is taken outside the summation, the quantity $b$ in this inequality
is always less than or equal to 1 resulting a non-negative $-\log(b)$ and therefore non-trivial. The additional term of von Neumann entropy is physically meaningful (though making the bound state-dependent) since a mixed state is expected to increase the uncertainty. By taking $N=2$, one simply recovers \eqref{MU}
with a tighter bound with an additional term $S(\rho )$. So we regard inequality (\ref{bound_MU}) as a generalization of MU inequality.
We remark that the MU inequality with term $S(\rho )$ can be obtained from the memory assisted entropic uncertainty
inequality \cite{Berta}, and we still call it the MU inequality in this Letter.

In addition, in the proof of Theorem \ref{generalization_MU}, we can obtain a corollary as a weighted uncertainty relation.
\begin{corollary}\label{weighted}
Suppose that we have three projective measurements $U$, $V$ and $W$ with bases $\{\ket{u_i}\}$, $\{\ket{v_j}\}$ and $\{\ket{w_k}\}$, we have
\begin{equation}
\begin{split}
&H(U)+H(V)+2H(W)\geq\\
&2S(\rho)-\log\{\max_{i,j,k}[c(u_i,w_k)c(w_k,v_j)]\}
\end{split}
\end{equation}
\end{corollary}
This is also a generalization of MU inequality.
But this inequality seems difficult to be extended to more observables.

\emph{Performance of the inequalities. }---Next, we shall show that
our result provides a non-trivial new bound for the uncertainty. Explicitly, we shall compare our new bound with known ones and show that ours is not overwhelmed.

To start, let us specify what other bounds will be considered. Note that one could always construct a multi-measurement inequality from two-measurement ones by summation. For instance, by simply combining three MU inequalities for mixed state,
\begin{equation}
H(U)+H(V)\geq -\log c(U,V)+S(\rho),
\end{equation}
we have the following inequality,
\begin{eqnarray}
&&H(M_1)+H(M_2)+H(M_3)-\frac{3}{2}S(\rho)\nonumber \\
&&\geq -\frac{1}{2}\log[c(M_1,M_2)c(M_2,M_3)c(M_3,M_1)].
\label{scb}
\end{eqnarray}
We will call the bounds constructed in this manner summation bounds hereafter.

Also, note that any two-measurement bound itself is a valid bound for multi-measurement cases. More precisely, if we have a bound $b(i,j)$ such that $H(M_i)+H(M_j)\geq b(i,j)$, we should also have $\sum_{m=1}^NH(M_m)\geq b(i,j)$.
This is straightforward, but we should note that two-measurement bounds are not necessarily lower than the summation bound mentioned above.
Therefore they are also needed to be taken into consideration.

For convenience, we call all summation bounds and two-measurement bounds the \emph{simply constructed bounds} (SCB), where we only consider the contribution of MU bounds from now on. We will later compare our result with the maximum among all SCBs.

Before running into numerical computation, we could first prove analytically that our bound is always no less than two-measurement MU bound. To see this, assume that we are to compare our result with the two-measurement bound for $M_1$ and $M_2$, which could always be achieved by a relabeling of the measurements. Then we have
\begin{eqnarray}
b&=&\max_{i_N}\{\sum_{i_2\sim i_{N-1}}\max_{i_1}[c(u^1_{i_1},u^2_{i_2})]\Pi_{m=2}^{N-1}c(u^m_{i_m},u^{m+1}_{i_{m+1}})\}
\nonumber \\
&\leq &\max_{i_N}\{\sum_{i_2\sim i_{N-1}}\max_{i_1,i_2}[c(u^1_{i_1},u^2_{i_2})]
\Pi_{m=2}^{N-1}c(u^m_{i_m},u^{m+1}_{i_{m+1}})\}
\nonumber \\
&=&\max_{i_1,i_2}c(u^1_{i_1},u^2_{i_2})=c(M_1,M_2).
\end{eqnarray}
Consequently, our bound is no less than two-measurement bound,
$-\log(b)\geq -\log c(M_1,M_2)$.
We remark that, in two-dimensional case, since the quantity $\max_{i}c(u^1_{i},u^2_{j})$
becomes exactly the same as $\max_{i,j}c(u^1_{i},u^2_{j})$, our bound actually reduces to two-measurement bound which is not very interesting.
However, this condition does not hold for higher dimensions.

Let us now consider an example for three measurements in three-dimensional space.
The measurements are chosen explicitly as follows,
\begin{equation}
\left\{
\begin{split}
&\{(1,0,0),(0,1,0),(0,0,1)\},\\
&\{(1/\sqrt{2},0,-1/\sqrt{2}),(0,1,0),(1/\sqrt{2},0,\sqrt{2})\},\\
&\{(\sqrt{a},e^{{\rm i}\phi}\sqrt{1-a},0),(\sqrt{1-a},-e^{{\rm i}\phi}\sqrt{a},0),(0,0,1)\}.
\end{split}
\right.
\end{equation}
The value of several bounds for $\phi=\pi/2$ with respect to $a$ is shown in Fig. \ref{comMU}.
Those bounds include the maximal SCB,
our bound of Theorem \ref{generalization_MU} and the RPZ direct sum majorization bound due to Rudnicki \emph{et al.} \cite{Karol}.
In this case, our bound is always better than the SCB and also becomes a complementary to the RPZ bound.
We thus confirm that our result is non-trivial.

\begin{figure}
\centering
\includegraphics[width=0.9\linewidth]{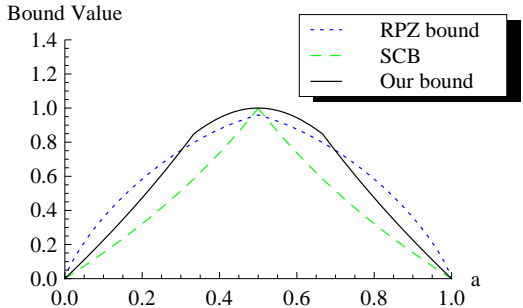}
\caption{(color online) Comparison of several bounds for $\phi=\pi/2$ with respect to $a$, including the maximal SCB in dashed-green line, our bound of Theorem \ref{generalization_MU} in solid-black line and the RPZ bound in dotted-blue line.
Here, $x$-axis indicates the value of $a$ and $y$-axis indicates the value of the bounds. }
\label{comMU}
\end{figure}

\emph{Entropic uncertainty relations in the presence of quantum memory. }---Recently, Berta \emph{et al.} \cite{Berta} introduced an entropic uncertainty relation in the presence of quantum memory as follows,
\begin{equation}
H(U|B)+H(V|B)\geq -\log c(U,V)+S(A|B)
\label{Bertabound}
\end{equation}
Here, $A$ and $B$ represent two particles in a two-body system $\rho_{AB}$,
which is generally mixed and might be entangled, $U$ and $V$ are two projective measurements applied on $A$, where $B$ is regarded as a quantum memory of system $A$. By definition, $H(U|B)$ is the conditional von Neumann entropy of the post-measurement state $\sum_i( \ket{u_i}\bra{u_i}\otimes I)\rho_{AB}(\ket{u_i}\bra{u_i}\otimes I)$, and $H(V|B)$ is similarly defined.
We know that the conditional entropy takes the form, $S(A|B)=S(\rho_{AB})-S(\rho_B)$. When $\rho_{AB}$ is pure,
we can find, $H(U|B)=H(U)-S(\rho_B)$, $H(V|B)=H(V)-S(\rho_B)$.

In quantum theory, the conditional entropy $S(A|B)$ can become negative,
implying that $A$ and $B$ are entangled \cite{adami-cerf}.
So this inequality shows that the existence of memory $B$ can help reduce uncertainty.
The conditional entropy represents also partial quantum information related with quantum
state merging \cite{winternature}.
An easier and more heuristic proof of this memory assisted uncertainty inequality
and its generalization to tripartite system,
with one being measured and two particles being used as memories, have been studied by Coles \emph{et al.} \cite{Coles}.

We could now, using our method, provide a generalization from a different view point
of the memory assisted uncertainty inequality to multi-measurement cases.
Our result is presented as follows,
\begin{theorem}\label{generalization_memory}
For a bipartite state $\rho_{AB}$ and $N$ projective measurements $\{M_i\}$ applied on $A$, we have
\begin{equation}
\sum_{m=1}^NH(M_m|B)\geq-\log(b)+(N-1)S(A|B),
\label{NBertabound}
\end{equation}
where $b$ is the same as that in Theorem \ref{generalization_MU}.
\end{theorem}
The proof is similar to that for Theorem \ref{generalization_MU}. All that we should do is to replace $\rho$ by $\rho_{AB}$
and go along the same process. On the other hand as shown in \cite{Berta}, we would like to remark that
if taking the dimension of $B$ to be zero, the uncertainty relation in the presence of quantum memory
can be reduced to case without quantum memory shown in (\ref{bound_MU}).

Actually, this is not the only possible formalism of multi-measurement memory-assisted inequality.
For example consider (\ref{bound_MU}), if we interpret $\rho$ there to be a subsystem $\rho_A$ of a \emph{pure} bipartite state $\rho_{AB}$
then with $S(\rho_A)=S(\rho_B)$, we have straightforwardly the result,
\begin{eqnarray}
\sum_{m=1}^NH(M_m)&\geq &-\log(b)+(N-1)S(\rho_A)\nonumber \\
&=&-\log(b)+NS(\rho_B)+S(A|B).
\end{eqnarray}
One can easily recover the following corollary after subtracting $NS(\rho_B)$ on both sides.
\begin{corollary}
For a bipartite pure state $\rho_{AB}$ and $N$ projective measurements $\{M_i\}$ applied on $A$, we have
\begin{equation}
\sum_{m=1}^NH(M_m|B)\geq-\log(b)+S(A|B).
\label{pure_NBertabound}
\end{equation}
\end{corollary}

We may also have a SCB for multiple measurements by repeatedly using the inequality
with two measurements (\ref{Bertabound}).
Similarly, all those bounds are complementary to each other. When $S(A|B)$ is negative, the SCB from
(\ref{Bertabound}) or the special case (\ref{pure_NBertabound}) for pure state can be tighter, when $S(A|B)$ is positive, the bound in (\ref{NBertabound}) is tighter.

\emph{Discussions.}---Entropic uncertainty relations for multiple measurements
are fundamental in quantum physics and can be applied for general
quantum key distribution protocols. We present the general uncertainty relations
for three different but related cases, the Deutsch type inequalities, the
MU type inequalities and the case in the presence of quantum memory.
Non-trivial and easy to compute bounds are presented which can provide a more precise description
of the uncertainty principle for quantum mechanics. The experimental
realization \cite{experiment1,experiment2} can also be implemented with more than two measurements settings.

\emph{Acknowledgement. }---We thank useful discussions with A. Winter, X. J. Ren and Y. C. Chang.
We thank K. Zyczkowski for correspondence.
This work was supported by NSFC (11175248), NFFTBS (J1030310, J1103205) and
grants from the Chinese Academy of Sciences.

\begin{widetext}

\appendix

\section{Generalization of Deutsch's inequality}\label{proofD}
We will here prove Theorem \ref{generalization_D} in the main article, which is not quite difficult. By the definition of R\'{e}nyi entropy, $H_{\infty}(\{p_i\})=-\log(\max_ip_i)$. Then the l.h.s. of the inequality is
\begin{eqnarray}
&&\sum_m H_{\infty}(M_m)=-\log(\max_{i_1}p^1_{i_1}\max_{i_2}p^2_{i_2}\ldots\max_{i_N}p^N_{i_N})\\
&&=\log[\max_{i_1,\ldots,i_N}(p^1_{i_1}p^2_{i_2}\ldots p^N_{i_N})]
\end{eqnarray}
Here we have used the similar notation convention as in the main article that superscripts represent the labels of measurements and subscripts represent the corresponding outcomes. To prove Theorem \ref{generalization_D}, it suffices to provide an upper bound on $\max_{i_1,\ldots,i_N}(p^1_{i_1}\ldots p^N_{i_N})$. We first do it for two measurements and show it could be easily generalized.

Recall the notation of measurements and bases in the main article. We tend to bound the term $\max_{i,j}|\langle u^1_i|\Phi\rangle\langle\Phi|u^2_j\rangle|$ with a larger and state-independent value, where we denote by $\ket{\Phi}$ the state of the measured system. Note that we are finally interested in the absolute value, thus we can convert our discussion from a complex vector space to a real one. More precisely, for a certain orthonormal basis $\{|k\rangle\}$ and any vector $|\eta\rangle=\sum_{k=1}^{d}\alpha_k|k\rangle$, define $|\tilde{\eta}\rangle=\sum_{k}|\alpha_k||k\rangle$. This map $|\eta\rangle\mapsto|\tilde{\eta}\rangle$ reduces this space to a subset of a real Euclidean space. Then, replacing every vector by its real image, all the upper bounds we get will also hold for original vectors, because obviously we have $|\langle a|b\rangle|=|a_1^*b_1+a_2^*b_2+...+a_n^*b_n|\leq |a_1||b_1|+|a_2||b_2|+...+|a_n||b_n|=\inner{\tilde{a}}{\tilde{b}}$. Here, the angle between two reduces vectors, which is always well defined in an Euclidean space, ranges from 0 to $\pi/2$.

Therefore we have
\begin{eqnarray}
&&|\langle u^1_i|\Phi\rangle\langle\Phi|u^2_j\rangle|=|\inner{u^1_i}{\Phi^{\|}}\inner{\Phi^{\|}}{u^2_j}|\\
&&\leq |\inner{\tilde{u^1_i}}{\tilde{\Phi^{\|}}}\inner{\tilde{\Phi^{\|}}}{\tilde{u^2_j}}|\\
&&=|\cos(\theta_1)\cos(\theta_2)|=\frac{1}{2}|\cos(\theta_1+\theta_2)+\cos(\theta_1-\theta_2)|\\
&&\leq\frac{1}{2}|\cos(\theta_1+\theta_2)+1|\\
&&=\frac{1}{2}(1+|\langle\tilde{u^1_i}|\tilde{u^2_j}\rangle|),
\end{eqnarray}
where $\ket{\Phi^{\|}}$ is the projection of $\ket{\Phi}$ onto the plane expanded by $\ket{u^1_i}$ and $\ket{u^2_j}$, $\theta_1$ is the angle between $\ket{\tilde{u^1_i}}$ and $\ket{\tilde{\Phi^{\|}}}$, and $\theta_2$ is that between $\ket{\tilde{u^2_j}}$ and $\ket{\tilde{\Phi^{\|}}}$. Then notice that, we could always choose a certain basis $\{\ket{k}\}$ with which we define the reduction such that $\ket{\tilde{u^1_i}}=\ket{u^1_i}$ and $\ket{\tilde{u^2_j}}={\rm e}^{{\rm i} \alpha}\ket{u^2_j}$. Therefore we have
\begin{equation}
|\langle u^1_i|\Phi\rangle\langle\Phi|u^2_j\rangle|\leq \frac{1}{2}(1+|\langle\tilde{u^1_i}|\tilde{u^2_j}\rangle|)=\frac{1}{2}(1+|\langle u^1_i|u^2_j\rangle|).
\end{equation}
If taking maximum here over $i$ and $j$, one will recover the result of Deustch \cite{Deutsch}, but we can go further.

We have straightforwardly that
\begin{eqnarray}
&&=|\langle u_{i_1}^1|\Phi\rangle\langle u_{i_2}^2|\Phi\rangle...\langle u_{i_N}^N|\Phi\rangle|\\
&&=\sqrt{|\langle u_{i_1}^1|\Phi\rangle\langle u_{i_2}^2|\Phi\rangle|}
\sqrt{|\langle u_{i_2}^2|\Phi\rangle\langle u_{i_3}^3|\Phi\rangle|}...
\sqrt{|\langle u_{i_{N-1}}^{N-1}|\Phi\rangle\langle u_{i_N}^N|\Phi\rangle|}
\sqrt{|\langle u_{i_N}^N|\Phi\rangle\langle u_{i_1}^1|\Phi\rangle|}\\
&&\leq \sqrt{(\frac{1+\cos\theta_{1,2}}{2})}\sqrt{(\frac{1+\cos\theta_{2,3}}{2})}...\sqrt{(\frac{1+\cos\theta_{N,N-1}}{2})}
\sqrt{(\frac{1+\cos\theta_{N,1}}{2})}\\
&&=\sqrt{\prod_{m=1}^N(\frac{1+\theta_{m,m+1}}{2})}=\sqrt{\prod_{m=1}^{N}(\frac{1+|\langle u_{i_m}^m|u_{i_{m+1}}^{m+1}\rangle|}{2})}.
\end{eqnarray}
Again, $N+1$ in the superscripts or subscripts is equivalent to 1. We square the inequality above, take maximum and logarithm, then it is exactly the desired result.

\section{Generalization of the Inequality of Maassen and Uffink}\label{proofMU}
We will here prove Theorem \ref{generalization_MU} in the main article as well as its corollaries. Considering the complexity on notation, we will first provide the proof for a simplified condition where there are only three measurements, and then generalize it by induction.
\subsection{Entropic uncertainty relation for three measurements}\label{3measurements}
For simplicity, we use the abbreviation $[\psi]$ to denote the projector $\ket{\psi}\bra{\psi}$. In the proof of our result, we utilize the concept of quantum relative entropy: by definition, $S(\rho||\sigma):=\Tr(\rho\log\rho)-\Tr(\rho\log\sigma)$. The method of our proof is inspired by Ref. \cite{Coles}.\\
\indent Consider three projective measurements, namely $U=\{\ket{u_i}\}$, $V=\{\ket{v_j}\}$, $W=\{\ket{w_k}\}$.
We have the following inequality.
\begin{theorem}
Denote by H($\cdot$) the Shannon entropy of a set of probabilities of a measurement, we have an entropic uncertainty relation:
\begin{equation}
H(U)+H(V)+H(W)\geq -\log(\max_i\sum_k(\max_j c(v_j,w_k))c(w_k,u_i))+2S(\rho),
\end{equation}
where $c(a,b):=|\langle a_i|b_j\rangle|^2$, and $S(\rho)$ is the von Neumann entropy of the state being measured.
\end{theorem}

\begin{proof}
First, we notice the following relation:
\begin{eqnarray}
&&S(\rho||\sum_j[v_j]\rho[v_j])\\
&&=\Tr(\rho\log\rho)-\Tr(\rho\log(\sum_j[v_j]\rho[v_j]))\\
&&=-S(\rho)-\Tr(\rho\log(\sum_j\ket{v_j}p_j\bra{v_j}))~~~(p_j:=\bra{v_j}\rho\ket{v_j})\\
&&=-S(\rho)-\Tr(\rho(\sum_j\ket{v_j}\log p_j\bra{v_j}))\\
&&=-S(\rho)-\sum_j(\Tr(\rho\ket{v_j}\log p_j\bra{v_j}))\\
&&=-S(\rho)-\sum_j p_j\log p_j\\
&&=-S(\rho)+H(V).
\end{eqnarray}
We then have
\begin{eqnarray}
&&-S(\rho)+H(V)=S(\rho||\sum_j[v_j]\rho[v_j])\\
&&\geq S(\sum_k[w_k]\rho[w_k]||\sum_{j,k}\ket{w_k}p_jc(w_k,v_j)\bra{w_k})~~~\text{(explained later)}\\
&&=\Tr(\sum_k[w_k]\rho[w_k]\log(\sum_{k}[w_k]\rho[w_k]))\\
&&-\Tr(\sum_k[w_k]\rho[w_k]\log(\sum_{j,k}\ket{w_k}p_jc(w_k,v_j)\bra{w_k}))\\
&&=-H(W)-\Tr(\sum_k[w_k]\rho[w_k]\log(\sum_k\ket{w_k}\alpha_k\bra{w_k}))\\
&&=-H(W)-\Tr(\rho\log(\sum_k\ket{w_k}\alpha_k\bra{w_k}))\\
&&=-H(W)+S(\rho||\sum_{j,k}\ket{w_k}p_jc(w_k,v_j)\bra{w_k})+S(\rho).
\end{eqnarray}
The second line invoked $S(\rho||\sigma)\geq S(\mathcal{E}(\rho)||\mathcal{E}(\sigma))$ (see Page 208 of Ref. \cite{Vedral}) with $\mathcal{E}(\rho)=\sum_k[w_k]\rho[w_k]$. \\
Thus we first get
\begin{equation}
H(V)+H(W)\geq S(\rho||\sum_{j,k}\ket{w_k}p_jc(w_k,v_j)\bra{w_k})+2S(\rho).
\label{2_ob_dep}
\end{equation}
Then, we have
\begin{eqnarray}
&&-2S(\rho)+H(V)+H(W)\geq S(\rho||\sum_{j,k}\ket{w_k}p_jc(w_k,v_j)\bra{w_k})\\
&&\geq S(\sum_i[u_i]\rho[u_i]||\sum_{i,j,k}\ket{u_i}p_jc(w_k,v_j)c(w_k,u_i)\bra{u_i})~~~\text{(again)}\\
&&=S(\sum_i[u_i]\rho[u_i]||\sum_{i}\ket{u_i}\beta_i\bra{u_i})~~~(\beta_i:=\sum_{j,k}p_jc(w_k,v_j)c(w_k,u_i))\\
&&\geq S(\sum_i[u_i]\rho[u_i]||\sum_{i}\ket{u_i}\max_i(\beta_i)\bra{u_i})\\
&&=S(\sum_i[u_i]\rho[u_i]||hI)~~~(h:=\max_i \beta_i)\\
&&=-H(U)-\Tr(\sum_i[u_i]\rho[u_i]\log(hI))\\
&&=-H(U)-\log(h)\cdot \Tr(I\sum_i[u_i]\rho[u_i])\\
&&=-H(U)-\log(h).
\end{eqnarray}
We get
\begin{equation}
H(U)+H(V)+H(W)\geq -\log(h)+2S(\rho).
\end{equation}
However, the term $h$ ($h=\max_i\sum_{j,k}p_jc(w_k,v_j)c(w_k,u_i)$) is still state-dependent. To obtain a state-independent bound, we should take maximum over $j$ \emph{inside} the summation. More precisely,
\begin{eqnarray}
&&h=\max_i\sum_{j,k}p_jc(w_k,v_j)c(w_k,u_i)\\
&&\leq \max_i\sum_{j,k}p_j(\max_j c(v_j,w_k))c(w_k,u_i)\\
&&=\max_i\sum_k(\max_j c(v_j,w_k))c(w_k,u_i).
\end{eqnarray}
Finally,
\begin{equation}
H(U)+H(V)+H(W)\geq -\log(b)+2S(\rho),
\end{equation}
where $b=\max_i\sum_k(\max_j c(v_j,w_k))c(w_k,u_i)$.\\
\end{proof}
\indent We note that since the sum over $i$ in the expression of $b$ is \emph{outside} the summation over $k$, $b$ is always lower than or equal to 1, thus our bound is non-trivial.

Actually, from the proof of this theorem, we can further get a corollary that is a weighted uncertainty relation. Recall \eqref{2_ob_dep}:
\begin{eqnarray}
&&H(V)+H(W)\geq S(\rho||\sum_{j,k}\ket{w_k}p_jc(w_k,v_j)\bra{w_k})+2S(\rho)\\
&&=S(\rho)-\Tr(\rho\log(\sum_k\ket{w_k}[\sum_jp_jc(w_k,v_j)]\bra{w_k})),
\end{eqnarray}
so obviously we have:
\begin{eqnarray}
&&H(W)+H(V)+H(W)+H(U)\\
&&\geq S(\rho)-\Tr(\rho\log(\sum_k\ket{w_k}[\sum_jp_jc(w_k,v_j)]\bra{w_k}))\\
&&+S(\rho)-\Tr(\rho\log(\sum_k\ket{w_k}[\sum_iq_ic(w_k,u_i)]\bra{w_k}))~~~(q_i:=\bra{u_i}\rho\ket{u_i})\\
&&=2S(\rho)-\Tr(\rho\log(\sum_k\ket{w_k}[\sum_{i,j}p_jq_ic(w_k,u_i)c(w_k,v_j)]\bra{w_k}))\\
&&\geq 2S(\rho)-\Tr(\rho\log(\max_{i,j,k}[c(u_i,w_k)c(w_k,v_j)]I))\\
&&=2S(\rho)-\log(\max_{i,j,k}(c(u_i,w_k)c(w_k,v_j))).
\end{eqnarray}
This is actually Corollary \ref{weighted} in the main article.
\subsection{Arbitrary number of measurements}
We boil down the proof into two steps.
\begin{lemma}\label{lemma1}
Given $N$ measurements $M_1,~M_2,~...M_N$, we have
\begin{equation}
-NS(\rho)+\sum_{m=1}^NH(M_m)\geq S(\rho||\sum_j\ket{u^N_j}\beta^N_j\bra{u^N_j}),
\end{equation}
where $\beta^N_j:=\sum_{i_1,i_2,...i_{N-1}} c(\rho,u^1_{i_1})c(u^1_{i_1},u^2_{i_2})...c(u^{N-1}_{i_{N-1}},u^N_j)$. We have used the consistent notation $c(\rho,u^m_i):=\bra{u^m_i}\rho\ket{u^m_i}=p^m_i$.
\end{lemma}
\begin{proof}
We have got this relation for $N=2$ (see equation \eqref{2_ob_dep}), so we prove it by induction. Suppose that this relation is satisfied for $N$ measurements, we proceed to prove it for $N+1$. We have
\begin{eqnarray}
&&-NS(\rho)+\sum_{m=1}^NH(M_m)\\
&&\geq S(\sum_k[u^{N+1}_k]\rho[u^{N+1}_k]||\sum_{j,k}\ket{u^{N+1}_k}\beta^N_jc(u^N_j,u^{N+1}_k)\bra{u^{N+1}_k})\\
&&=-H(M_{N+1})-\Tr[\sum_k[u^{N+1}_k]\rho[u^{N+1}_k]\log(\sum_k\ket{u^{N+1}_k}\eta_k\bra{u^{N+1}_k})]\\
&&=-H(M_{N+1})-\Tr[\sum_k[u^{N+1}_k]\rho[u^{N+1}_k]\sum_l\ket{u^{N+1}_l}\log\eta_l\bra{u^{N+1}_l}]\\
&&=-H(M_{N+1})-\sum_k\Tr([u^{N+1}_k]\rho[u^{N+1}_k]\sum_l\ket{u^{N+1}_l}\log\eta_l\bra{u^{N+1}_l})\\
&&=-H(M_{N+1})-\sum_k\Tr(\rho\sum_l\ket{u^{N+1}_k}\delta_{kl}\log{\eta_l}\delta_{kl}\bra{u^{N+1}_k})\\
&&=-H(M_{N+1})-\sum_k\Tr(\rho\ket{u^{N+1}_k}\log{\eta_k}\bra{u^{N+1}_k})\\
&&=-H(M_{N+1})-\Tr(\rho\log(\sum_k\ket{u^{N+1}_k}\eta_k\bra{u^{N+1}_k}))+\Tr(\rho\log\rho)-\Tr(\rho\log\rho)\\
&&=-H(M_{N+1})+S(\rho||\sum_k\ket{u^{N+1}_k}\eta_k\bra{u^{N+1}_k})+S(\rho)
\end{eqnarray}
Notice that $\eta_k=\sum_j\beta^N_jc(u^N_j,u^{N+1}_k)=\beta^{N+1}_k$ (second and third lines). We therefore finish the proof.
\end{proof}
Just like the previous proof, if we impose a state-independent as well as $j$-independent upper bound to $\beta^N_j$, we can get a state-independent uncertainty relation.

Then we simply have
\begin{eqnarray}
&&-NS(\rho)+\sum_{m=1}^NH(M_m)\geq S(\rho||\sum_j\ket{u^N_j}\beta^N_j\bra{u^N_j})\\
&&=-S(\rho)-\Tr(\rho\log(\sum_j\ket{u^N_j}\beta^N_j\bra{u^N_j}))\\
&&\geq -S(\rho)-\Tr(\rho\log(bI))~~~\text{(since $b$ is greater than $\beta^N_j$)}\\
&&=-S(\rho)-\log(b).
\end{eqnarray}

This is nothing but Theorem \ref{generalization_MU} in the main article.
\section{Uncertainty relation with quantum memory}\label{conditionalentropy}
To start, we justify some assertions about quantum conditional entropy mentioned in the main article.

By the definition of Berta \emph{et al.} \cite{Berta}, $H(U|B)$ is the conditional von Neumann entropy of the post-measurement state $\sum_i( \ket{u_i}\bra{u_i}\otimes I)\rho_{AB}(\ket{u_i}\bra{u_i}\otimes I)$, where $U$ is a projective measurement performed on system $A$. An alternative definition is used in the work of Coles \emph{et al.} \cite{Coles}:
\begin{equation}
H(U|B):=H(U)-\chi(U,B),
\end{equation}
where $\chi(U,B)$ is the Holevo quantity $S(\rho_B)-\sum_jp_jS(\rho_{B,j})$. $\rho_{B,j}$ is the state of $B$ when measurement $U$ gets the $j$th outcome (with probability $p_j$), i.e. $\rho_{B,j}=\Tr_A(\ket{u_j}\bra{u_j}\rho_{AB})/p_j$.

To prove the equivalence of the two definitions, we utilize Lemma 1 in Ref. \cite{Colespra} (see also Page 513 of Ref. \cite{Nielsen&Chuang}), which says that
\begin{equation}
S(\sum_jp_j\rho_j)-\sum_jp_jS(\rho_j)\leq H(\{p_j\}),
\end{equation}
with equality if and only if $\rho_j$ are mutually orthogonal.

We then have
\begin{eqnarray}
&&S(\sum_i( \ket{u_i}\bra{u_i}\otimes I)\rho_{AB}(\ket{u_i}\bra{u_i}\otimes I))-S(\rho_B)\\
&&=S(\sum_jp_j\ket{u_j}\bra{u_j}\otimes\rho_{B,j})-S(\rho_B)\\
&&=H(U)+\sum_jp_jS(\rho_{B,j})-S(\rho_B)\\
&&=H(U)-\chi(U,B).
\end{eqnarray}
When $\rho_{AB}$ is pure, we can always expand the state as $\ket{\psi_{AB}}=\sum_{i}\alpha_i\ket{u_i}\ket{\psi_{B,i}}$, so $\rho_{B,i}$ is also a pure state. Therefore, in this special case, $H(U|B)=H(U)-S(\rho_B)$, which is used in the main article.

Now, let's begin our process to generalize the result of Berta \emph{et al.} The idea is straightforward: we just replace $\rho$ by $\rho_{AB}$ in our proof in Appendix \ref{proofMU} and see what happens. There is no doubt that most of the calculation will be similar, so we just note some key points here.

First, recall that we made a connection between information entropy of a measurement with quantum relative entropy, i.e. $S(\rho||\sum_j[v_j]\rho[v_j])=-S(\rho)+H(V)$. As for bipartite systems $\rho_{AB}$, we have
\begin{eqnarray}
&&S(\rho_{AB}||\sum_j[v_j]\rho_{AB}[v_j])\\
&&=-S(\rho_{AB})-\Tr(\rho_{AB}\log(\sum_j[v_j]\rho_{AB}[v_j]))\\
&&=-S(\rho_{AB})+S(\sum_j[v_j]\rho_{AB}[v_j])\\
&&=-S(\rho_{AB})+S(\rho_B)+S(\sum_j[v_j]\rho_{AB}[v_j])-S(\rho_B)\\
&&=H(V|B)-S(A|B).
\end{eqnarray}
Recollect that $[v_j]$ is the measurement projector acting only on system $A$. This results are obviously similar.

Then, our standard method with which the three-measurement inequality is proven can be performed in the same manner.
\begin{eqnarray}
&&S(\rho_{AB}||\sum_j[v_j]\rho_{AB}[v_j])\\
&&\geq S(\sum_k[w_k]\rho_{AB}[w_k]||\sum_{j,k}\ket{w_k}c(w_k,v_j)\bra{v_j}\rho_{AB}\ket{v_j}\bra{w_k})\\
&&=-S(\sum_k[w_k]\rho_{AB}[w_k])-\Tr(\sum_k[w_k]\rho_{AB}[w_k]\log(\sum_{j,k}\ket{w_k}c(w_k,v_j)\bra{v_j}\rho_{AB}\ket{v_j}\bra{w_k}))\\
&&=-S(\sum_k[w_k]\rho_{AB}[w_k])-\Tr(\rho_{AB}\log(\sum_{j,k}\ket{w_k}c(w_k,v_j)\bra{v_j}\rho_{AB}\ket{v_j}\bra{w_k}))\\
&&=-S(\sum_k[w_k]\rho_{AB}[w_k])+S(\rho_{AB}||\sum_{j,k}\ket{w_k}c(w_k,v_j)\bra{v_j}\rho_{AB}\ket{v_j}\bra{w_k})+S(\rho_{AB})\\
&&=-H(W|B)+S(\rho_{AB}||\sum_{j,k}\ket{w_k}c(w_k,v_j)\bra{v_j}\rho_{AB}\ket{v_j}\bra{w_k})+S(A|B).
\end{eqnarray}
Here, the term $\bra{v_j}\rho_{AB}\ket{v_j}$ is simply $p_j\rho_{B,j}$. Hence, the result above is again much similar to our previous result---one need only replace $H(\cdot)$ by $H(\cdot|B)$, $S(\rho)$ by $S(A|B)$, and $p_j$ by $p_j\rho_{B,j}$. Consequently, a generalized lemma corresponding to Lemma \ref{lemma1} can be obtained easily.
\begin{lemma}
Given a bipartite state $\rho_{AB}$ and N projective measurements $M_1,~M_2,~...M_N$ acting on system $A$, we have
\begin{equation}
-NS(A|B)+\sum_{m=1}^NH(M_m|B)\geq S(\rho_{AB}||\sum_j\ket{u^N_j}\beta^N_j\bra{u^N_j}),
\end{equation}
where $\beta^N_j:=\sum_{i_1,i_2,...i_{N-1}} c(\rho_{AB},u^1_{i_1})c(u^1_{i_1},u^2_{i_2})...c(u^{N-1}_{i_{N-1}},u^N_j)$. We have used the consistent notation $c(\rho_{AB},u^m_i):=\bra{u^m_i}\rho_{AB}\ket{u^m_i}=p^m_i\rho_{B,m_i}$.
\end{lemma}
We further reduce the r.h.s. to a state-independent bound.
\begin{eqnarray}
&&S(\rho_{AB}||\sum_j\ket{u^N_j}\beta^N_j\bra{u^N_j})\\
&&=S(\rho_{AB}||\sum_j\ket{u^N_j}[\sum_{i_1,i_2,...i_{N-1}} c(\rho_{AB},u^1_{i_1})c(u^1_{i_1},u^2_{i_2})...c(u^{N-1}_{i_{N-1}},u^N_j)]\bra{u^N_j})\\
&&\geq S(\rho_{AB}||\sum_j\ket{u^N_j}[\sum_{i_1,i_2,...i_{N-1}}\bra{u^1_{i_1}}\rho_{AB}\ket{u^1_{i_1}}\max_{i_1}(c(u^1_{i_1},u^2_{i_2}))...c(u^{N-1}_{i_{N-1}},u^N_j)]\bra{u^N_j})\\
&&=S(\rho_{AB}||\sum_j\ket{u^N_j}[\sum_{i_2,...i_{N-1}}\rho_B\max_{i_1}(c(u^1_{i_1},u^2_{i_2}))...c(u^{N-1}_{i_{N-1}},u^N_j)]\bra{u^N_j})\\
&&\geq S(\rho_{AB}||\sum_j\ket{u^N_j}[b\rho_{B}]\bra{u^N_j})\\
&&=S(\rho_{AB}||bI\otimes\rho_B)\\
&&=-S(A|B)-\log(b).
\end{eqnarray}
A justification of the last step can be seen in equations 11.104-11.106 on Page 521 of Ref. \cite{Nielsen&Chuang}. Therefore we have proven that
\begin{equation}
\sum_{m=1}^NH(M_m|B)\geq -\log(b)+(N-1)S(A|B).
\end{equation}
This is exactly Theorem \ref{generalization_memory} in the main article.

\end{widetext}

\end{document}